\pdfoutput=1
\documentclass[a4paper,UKenglish,cleveref]{lipics-v2021}%
\nolinenumbers

\usepackage[scr=boondoxo,scrscaled=1.05]{mathalpha}

\usepackage[utf8]{inputenc}
\usepackage[T1]{fontenc}
\usepackage{microtype}
\renewcommand{\phi}{\varphi}
\renewcommand{\epsilon}{\varepsilon}
\usepackage[]{algorithm2e}

\usepackage{comment}

\usepackage{lipsum}

\usepackage{amsmath}
\usepackage{amssymb}
\usepackage{amsthm}
\usepackage{mathtools}
\usepackage{todonotes}
\let\tempcup\cup
\let\tempcap\cap
\usepackage{mathabx}
\let\cup\tempcup
\let\cap\tempcap
\usepackage{graphicx}

\catcode`\@ = 11
\newdimen\@InsertBoxMargin
\newcount\@numlines    %
\newcount\@linesleft   %
\def\ParShape{%
    \@numlines = 0
    \def\@parshapedata{ }%
    \afterassignment\@beginParShape
    \@linesleft
}%
\def\@beginParShape{%
    \ifnum \@linesleft = 0
      \let\@whatnext = \@endParShape
    \else
      \let\@whatnext = \@readnextline
    \fi
    \@whatnext
}%
\def\@endParShape{%
    \global\parshape = \@numlines \@parshapedata
}%
\def\@readnextline#1 #2 #3 {%
    \ifnum #1 > 0
      \bgroup  %
        \dimen0 = \hsize
        \advance \dimen0 by -#2  %
        \advance \dimen0 by -#3  %
        \count0 = 0
        \loop
          \global\edef\@parshapedata{%
            \@parshapedata    %
            #2                %
            \space            %
            \the\dimen0       %
            \space            %
          }%
          \advance \count0 by 1
          \ifnum \count0 < #1
        \repeat
      \egroup
      \advance \@numlines by #1
    \fi
    \advance \@linesleft by -1
    \@beginParShape
}%
\newbox\@boxcontent     %
\newcount\@numnormal    %
\newdimen\@framewidth   %
\newdimen\@wherebottom  %
\newif\if@byframe       %
\@byframefalse
\def\InsertBoxC#1{%
  \leavevmode
  \vadjust{
    \vskip \@InsertBoxMargin
    \hbox to \hsize{\hss#1\hss}
    \vskip \@InsertBoxMargin
  }%
}%
\def\InsertBoxL#1#2{%
  \@numnormal = #1
  \setbox\@boxcontent = \hbox{#2}%
  \let\@side = 0
  \futurelet \@optionalparameter \@InsertBox
}
\def\InsertBoxR#1#2{%
  \@numnormal = #1
  \setbox\@boxcontent = \hbox{#2}%
  \let\@side = 1
  \futurelet \@optionalparameter \@InsertBox
}%
\def\@InsertBox{%
  \ifx \@optionalparameter [
    \let\@whatnext = \@@InsertBoxCorrection
  \else
    \let\@whatnext = \@@InsertBoxNoCorrection
  \fi
  \@whatnext
}%
\def\@@InsertBoxCorrection[#1]{%
  \ifx \@side 0
    \@@InsertBox{#1}{0}{{\the\@framewidth} 0cm}%
  \else
    \@@InsertBox{#1}{1}{0cm {\the\@framewidth}}%
  \fi
}%
\def\@@InsertBoxNoCorrection{%
  \@@InsertBoxCorrection[0]%
}%
\def\@@InsertBox#1#2#3{%
  \MoveBelowBox
  \@byframetrue
  \@wherebottom = \baselineskip
  \multiply \@wherebottom by \@numnormal
  \advance \@wherebottom by 2\@InsertBoxMargin
  \advance \@wherebottom by \ht\@boxcontent
  \advance \@wherebottom by \pagetotal
  \ifdim \pagetotal = 0cm
    \advance \@wherebottom by -\baselineskip  %
  \fi
  \advance \@wherebottom by #1\baselineskip
  \@framewidth = \wd\@boxcontent
  \advance \@framewidth by \@InsertBoxMargin
  \bgroup  %
    \ifdim \pagetotal = 0cm
      \dimen0 = \vsize
    \else
      \dimen0 = \pagegoal
    \fi
    \ifdim \@wherebottom > \dimen0
      \immediate\write16{+--------------------------------------------------------------+}%
      \immediate\write16{| The box will not fit in the page. Please, re-edit your text. |}%
      \immediate\write16{+--------------------------------------------------------------+}%
      \vrule width \overfullrule
    \fi
  \egroup
  \prevgraf = 0
  \vbox to 0cm{%
    \dimen0 = \baselineskip
    \multiply \dimen0 by \@numnormal
    \advance \dimen0 by -\baselineskip
    \setbox0 = \hbox{y}%
    \vskip \dp0
    \vskip \dimen0
    \vskip \@InsertBoxMargin
    \ifnum #2 = 1
      \vtop{\noindent \hbox to \hsize{\hss \box\@boxcontent}}%
    \else
      \vtop{\noindent \box\@boxcontent}%
    \fi
    \vss
  }%
  \vglue -\parskip
  \vskip -\baselineskip
  \everypar = {%
    \ifdim \pagetotal < \@wherebottom
      \bgroup  %
        \dimen0 = \@wherebottom
        \advance \dimen0 by -\pagetotal
        \divide \dimen0 by \baselineskip
        \count1 = \dimen0
        \advance \count1 by 1
        \advance \count1 by -\@numnormal
        \ifnum #2 = 1
          \ParShape = 3
                      {\the\@numnormal}   0cm   0cm
                      {\the\count1}       0cm   {\the\@framewidth}
                      1                   0cm   0cm
        \else
          \ParShape = 3
                      {\the\@numnormal}   0cm                  0cm
                      {\the\count1}       {\the\@framewidth}   0cm
                      1                   0cm                  0cm
        \fi
      \egroup
    \else
      \@restore@    %
    \fi
  }%
  \def\par{%
      \endgraf
      \global\advance \@numnormal by -\prevgraf
      \ifnum \@numnormal < 0
        \global\@numnormal = 0
      \fi
      \prevgraf = 0
  }%
}%
\def\MoveBelowBox{%
  \par
  \if@byframe
    \global\advance \@wherebottom by -\pagetotal
    \ifdim \@wherebottom > 0cm
      \vskip \@wherebottom
    \fi
    \@restore@
  \fi
}%
\def\@restore@{%
    \global\@wherebottom = 0cm
    \global\@byframefalse
    \global\everypar = {}%
    \global\let \par = \endgraf
    \global\parshape = 1 0cm \hsize
}%
\ifx \documentclass \@Dont@Know@What@It@Is@
\else
  \let \pageno = \c@page
\fi

\catcode`\@ = 12
 \usepackage{bm}
\usepackage{bbm}

\usepackage{calc}
\usepackage{float}
\usepackage{booktabs}
\usepackage{graphicx}
\usepackage{tikz}
\usetikzlibrary{positioning,arrows.meta,trees,shapes,graphs, graphs.standard,decorations.pathmorphing, calc}

\newcommand{\N}{\mathbb{N}}
\newcommand{\bigO}{\ensuremath{\mathcal{O}}}
\newcommand{\assign}{\ensuremath{\mathscr{I}}}

\newcommand{\restrict}[1]{\raisebox{-.2ex}{$|$}_{#1}}

\DeclareMathOperator{\out}{out}
\renewcommand{\hat}{\widehat}
\usepackage{xcolor} %
\usepackage[most]{tcolorbox}

\usepackage{hyperref}
\usepackage[capitalise,noabbrev]{cleveref}
\usepackage{lipsum}

\author{Christoph Berkholz}{Technische Universität Ilmenau, Germany}{christoph.berkholz@tu-ilmenau.de}{https://orcid.org/0000-0002-3554-517X
}{Funded by the Deutsche Forschungsgemeinschaft (DFG, German
Research Foundation) – project number 414325841.}%
\author{Matthäus Micun}{Technische Universität Ilmenau, Germany}{matthaeus.micun@tu-ilmenau.de}{https://orcid.org/0009-0003-6658-8927}{Funded by the Deutsche Forschungsgemeinschaft (DFG, German
Research Foundation) – project number 414325841.}%

\author{Igor Razgon}{Durham University, United Kingdom}{igor.razgon@gmail.com}{https://orcid.org/0000-0002-7060-5780}{}%

\authorrunning{C. Berkholz, M. Micun I. Razgon}

\begin{document}

\pagenumbering{arabic}
\date{\today}

\title{Restructuring Tree Decision Diagrams}

\Copyright{Christoph Berkholz, Matthäus Micun, Igor Razgon, Wim Van den Broeck}
\maketitle


\begin{abstract}
    Tree Decision Diagrams (TDDs) are a data structure recently introduced by Capelli et al. (SAT 2026). 
    They are structured along a vtree and the size of their canonical form lies between Ordered Binary Decision Diagrams (OBDDs) and deterministic structured DNNF circuits (d-SDNNFs). While the succintness gap between TDD and d-SDNNF is exponential, only a quasipolynomial separation between OBDD and TDD has been shown and it was left as open question whether this is optimal.  
    We answer this question affirmatively by showing that every TDD can be transformed to an equivalent OBDD of quasipolynomial size. 
    
    Although this might be seen as a weakness, our second result shows that TDDs share another desirable property with OBDDs that is not known to hold for d-SDNNF: Given a TDD and another target vtree, it is possible to construct the minimal and canonical TDD respecting the new vtree in time polynomial in the input and output. As a result we also obtain that the equivalence test between TDDs over different vtrees can be done in polynomial time.    
    \end{abstract}

\section{Introduction}

Decision diagrams and decomposable circuits are versatile data
structures for representing Boolean functions and several variants have been introduced and studied in \emph{knowledge compilation}
\cite{DBLP:journals/jair/DarwicheM02,DBLP:journals/corr/abs-2404-09674}.\footnote{The reader may also
  consult \url{https://circuitzoo.net/} for an up-to-date knowledge
  compilation map.}
Over the last years, these data structures and the ideas behind them
have also been used for different query evaluation algorithms, see
\cite{DBLP:journals/sigmod/AmarilliC24} for a survey on
knowledge compilation in database theory. 
In particular, \emph{factorized databases}
\cite{DBLP:journals/tods/OlteanuZ15} or \emph{relational circuits} \cite{DBLP:conf/icdt/Capelli26} use decomposable circuits to
succinctly represent query results. This approach is known to
match the best query evaluation algorithms for join queries
\cite{DBLP:conf/icalp/BerkholzV23,DBLP:conf/icdt/BerkholzV26} and is also
applicable to queries with negation \cite{DBLP:conf/icdt/CapelliI24}. Furthermore,
knowledge compilation formats that allow efficient model counting
have been used and studied for query evaluation on probabilistic
databases (known as \emph{query compilation} \cite{DBLP:journals/ftdb/BroeckS17}). Examples include  
ordered binary decision diagrams (OBDD)
\cite{DBLP:conf/icdt/JhaS12,DBLP:journals/mst/JhaS13,DBLP:journals/tods/FinkO16}, free
binary decision diagrams (FBDD)
and decision decomposable circuits (dec-DNNF) \cite{DBLP:journals/tods/Beame0RS17}, deterministic
structured DNNF (d-SDNNF) \cite{DBLP:conf/pods/BovaS17}, and deterministic
decomposable circuits (d-DNNF and d-D) \cite{DBLP:conf/amw/MonetO18,DBLP:conf/pods/Monet20}.
Decomposable circuits also have numerous applications for MSO query evaluation on trees
(see \cite{DBLP:journals/sigmod/AmarilliC24}).
Current developments in this field have recently been discussed at the Dagstuhl
seminar 26221 \emph{Knowledge Compilation in Artificial Intelligence,
  Databases, and Formal Methods} \cite{DagstuhlSeminar26221}.

In general there is a trade-off between usefulness and succinctness of
the knowledge compilation targets. At one end there are OBDDs
\cite{DBLP:journals/tc/Bryant86} which decide on the Boolean variables
along a variable order. They have the least compactness among the
formats mentioned so far but support many polynomial time queries such as
model counting and equivalence test as well as polynomial time
transformations such as negation and singleton forgetting
(existential quantification of a single Boolean variable). Moreover, OBDDs
have a canonical minimal form for a given variable order and minimizing an OBDD can also be done in polynomial time.

A generalization of OBDDs are \emph{deterministic structured decomposable circuit}s (d-SDNNF) \cite{darwichedsdnnf} which allow exponentially more succinct representations for some
Boolean functions \cite{bova2016sdds}. The ``structuredness'' means
that the variables are also ordered, but not along a linear order as
for OBDDs but by using a \emph{variable tree} (\emph{vtree}).
The succinctness of d-SDNNFs compared to OBDDs comes at a prize: it is not known, whether
equivalence between two d-SDNNF over different vtrees can be tested in polynomial time and
it has recently be shown, that they cannot be negated in polynomial
time and do not allow polynomial time singleton forgetting
\cite{Vinall-Smeeth_2024}. Already in 2017, Bova and Szeider introduced a restricted
\emph{canonical} form of d-SDNNF \cite[Lemma~4]{DBLP:conf/pods/BovaS17} and
analyzed their use for evaluating UCQs with and without inequalities
on probabilistic databases. Finally, they called for a deeper inspection:

\begin{quotation}
\emph{The canonical structured deterministic forms induced by
factorized implicants, introduced in Section 3.2.1, deserve in
our opinion both a direct investigation in the framework of
the knowledge compilation map, and a thorough comparison with the data structures used in factorized databases
} [\ldots]
\end{quotation}

The \emph{direct investigation} was only started recently: Capelli, Choi,
Mengel, Muñoz, and Van~den~Broeck gave a precise syntactic definition
of \emph{Tree Decision Diagrams (TDD)} \cite{tdd} that captures these
canonical forms and they analyzed them in depth.
In particular, they showed that TDDs share a number of desirable
properties with OBDDs such as polynomial time negation and singleton
forgetting. At the same time, TDDs of polynomial size are able to represent instances of
bounded treewidth, which is not possible for
OBDDs. As a consequence, it follows by the construction of Razgon \cite{razgon2014obdds} that TDDs are
at least quasipolynomially more succinct than OBDDs. On the other hand
they are exponentially less succinct than general d-SDNNFs and its subclass
\emph{sentential decision diagrams} (SDDs) \cite{tdd,bova2016sdds}.
As announced in \cite{DagstuhlSeminar26221}, TDDs are well suited for
bottom-up knowledge compilation and the implementation \emph{TiDiDi} has been quite
successful for a bottom-up solver at the 7th Model Counting
Competition (MC 2026) \cite{Competition2021_23}.\footnote{See \url{https://mccompetition.org/}.}

\subparagraph*{Our contribution} Our main contribution are two
independent results on TDDs that both ultimately rely on restructuring
the underlying vtree.

In \cite{tdd} the
authors obtained only a quasipolynomial succinctness gap between OBDDs and TDDs, whereas it is
known that other structured formats such as SDDs and d-SDNNF can be exponentially more succinct than
OBDDs \cite{bova2016sdds}. The authors stated as open problem, whether
this succinctness gap is optimal or if an exponential gap can be shown.
We answer this question exhaustively by showing that \emph{TDDs can be quasipolynomially
simulated by OBDDs} (Theorem~\ref{th:simulation}) and hence the
separation is optimal. To prove this
theorem, we perform a stepwise reordering of the vtree until it has a
linear form which then corresponds to an OBDD over a linear order.
The quasipolynomial blow-up then follows by a careful analysis of the
recurrence.

Although our simulation result might be seen as a weakness of TDDs,
it should be contrasted with the following related quasipolynomial
simulations of decomposable circuits by decision diagrams. In light of
these results, the property seems rather natural and might be a
further sign for the robustness of the new formalism:
\begin{itemize}
\item Decomposable negation normal form (DNNF) can be
  quasipolynomially simulated by nondeterministic read-once branching programs (nROBP)
    \cite{DBLP:conf/cp/Razgon15}.
\item Decision decomposable negation normal form (dec-DNNF) can be
  quasipolynomially simulated by read-once branching programs = free binary decision diagrams (FBDD)  \cite{BeameFBDDSimulation}.
\end{itemize}

Our second results shows that \emph{equivalence of TDDs over different
vtrees $T_1$ and $T_2$ can be checked in polynomial time}
(Lemma~\ref{lem:tdd-different-vtrees}). While this
is known to hold for OBDDs with different orders, it is currently not
known whether a polynomial time equivalence check is possible for other structured
formats such as SDDs or d-SDNNFs. To achieve
this result, we prove a slightly stronger result on reordering TDDs
that might be of independent interest:
there is an algorithm that receives a TDD $D_1$ over a vtree $T_1$ and
another vtree $T_2$ and computes an equivalent minimal TDD $D_2$ over
$T_2$ in time polynomial in the size of $D_1$ and $D_2$ (Theorem~\ref{th:restructuring}). As a nice consequence we also conclude that it is possible to check equivalence between TDDs and d-SDNNFs (Lemma~\ref*{lem:eqtdddsdnnf}).

\section{Preliminaries}\label{sec:prelim}

\begin{definition}[variable trees]
    A \emph{variable tree}, or just vtree, over a set of variables $X$ is a pair $(T,b)$, where T is a rooted full binary tree, and $b$ is a bijection from the leaves of T to $X$. For any $v \in V(T)$ we let $T_s$ be the subtree rooted at $s$, and $b(T_s)$ is the image of b restricted to the leaves of $T_s$.
\end{definition}

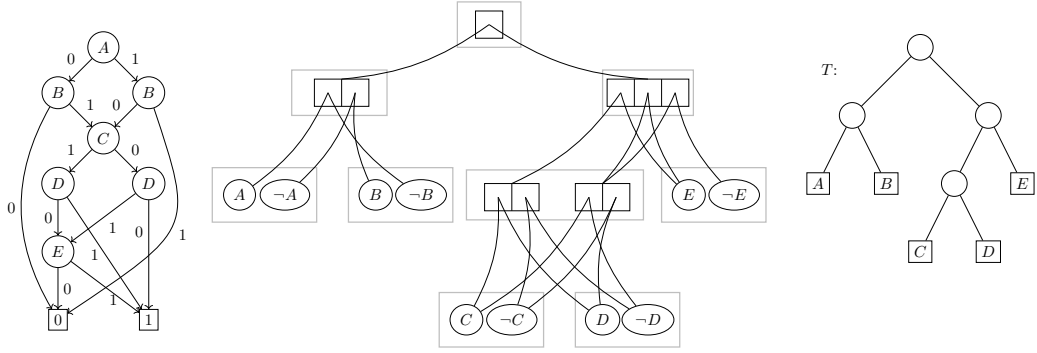
\begin{figure}
    \resizebox{\textwidth}{!}{
\tikzset{
  multi box/.style={
    draw,
    inner sep=0pt,
    outer sep=0pt
  },
  multi box divider/.style={
    draw
  },
}

\NewDocumentCommand{\multiboxnode}{O{} m m O{2} O{6mm}}{%
  \pgfmathsetlengthmacro{\multiboxtotalwidth}{#4*(#5)}%

  \node[
    multi box,
    #1,
    minimum width=\multiboxtotalwidth,
    minimum height=#5
  ] (#2) at #3 {};

  \foreach \i in {1,...,#4}{%
    \coordinate (#2-\i-center) at
      ($(#2.west)!{(\i-0.5)/#4}!(#2.east)$);
  }

  \ifnum#4>1\relax
    \foreach \i in {1,...,\number\numexpr#4-1\relax}{%
      \draw[multi box divider]
        ($(#2.north west)!{\i/#4}!(#2.north east)$)
        --
        ($(#2.south west)!{\i/#4}!(#2.south east)$);
    }
  \fi
} 
\tikzset{
  pics/leafnode/.style args={#1}{
    code={
      \draw[draw=gray!50, thick] (-1.1,-0.6) rectangle (1.2,0.6);
      \node[ellipse,draw] (-1) at (-0.5,0) {$#1$};
      \node[ellipse,draw] (-0) at (0.5,0) {$\neg #1$};
    }
  }
}

\begin{tikzpicture}
    [
        auto,
        dnnf/.style={ellipse,draw},
        vert/.style={circle,draw},
        sink/.style={rectangle,draw},
        vtree/.style={circle, inner sep=2mm,draw},
        bend angle=15
    ]

    \begin{scope}[
        scale=1,
        xshift=-8cm,
        yshift=4cm
    ]

    \node[vert] (root) at (0,0) {$A$};
    \node[vert] (Bleft) at (-1,-1) {$B$};
    \node[vert] (Bright) at (1,-1) {$B$};
    \node[vert] (C) at (0,-2) {$C$};

    \node[vert] (Dleft) at (-1,-3) {$D$};
    \node[vert] (Dright) at (1,-3) {$D$};

    \node[vert] (E) at (-1,-4.5) {$E$};

    \node[sink] (0sink) at (-1,-6) {$0$};
    \node[sink] (1sink) at (1,-6) {$1$};

    \draw[->] (root) -- node[swap]  {0} (Bleft);
    \draw[->] (root) -- node        {1} (Bright);
    \draw[->] (Bleft) -- node       {1} (C);
    \draw[->] (Bright) -- node[swap]{0} (C);

    \draw[->, bend right=30] (Bleft) to node[swap]  {0} (0sink);
    \draw[->, rounded corners] (Bright) .. node {1} controls (2,-4) .. (0sink);

    \draw[->] (C) -- node[swap]     {1} (Dleft);
    \draw[->] (C) -- node           {0} (Dright);

    \draw[->] (Dleft) -- node[swap] {0} (E);
    \draw[->] (Dleft) -- node[swap,yshift=2mm] {1} (1sink);

    \draw[->] (Dright) -- node      {1} (E);
    \draw[->] (Dright) -- node[swap, yshift=5mm]      {0} (1sink);

    \draw[->] (E) -- node[swap, xshift=4mm, yshift=-0.5mm] {1} (1sink);
    \draw[->] (E) -- node {0} (0sink);

    \end{scope}

    \begin{scope}[
    ]

    \pic (leaf1) at (0.5,-2) {leafnode={C}};
    \pic (leaf2) at (3.5,-2) {leafnode={D}};

    \draw[draw=gray!50,thick] (0,0.2) rectangle (3.9,1.3);
    \multiboxnode{i1}{(3,0.7)}[2][6mm]

    \draw[bend right] (leaf1-1) to (i1-1-center);
    \draw[bend left] (leaf2-0) to (i1-1-center);

    \draw[bend right] (leaf1-0) to (i1-2-center);
    \draw[bend left] (leaf2-1) to (i1-2-center);

    \multiboxnode{i2}{(1,0.7)}[2][6mm]

    \draw[bend right] (leaf1-1) to (i2-1-center);
    \draw[bend left] (leaf2-1) to (i2-1-center);

    \draw[bend right] (leaf1-0) to (i2-2-center);
    \draw[bend left] (leaf2-0) to (i2-2-center);

    \pic (leaf3) at (5.4,0.75) {leafnode={E}};

    \multiboxnode{i3}{(4,3)}[3]
    \draw[draw=gray!50,thick] (3,2.5) rectangle (5,3.5);

    \draw[bend right] (i1.north) to (i3-3-center);
    \draw[bend left] (leaf3-0) to (i3-3-center);

    \draw[bend right] (i1.north) to (i3-2-center);
    \draw[bend left] (leaf3-1) to (i3-2-center);

    \draw[bend right] (i2.north) to (i3-1-center);
    \draw[bend left] (leaf3-1) to (i3-1-center);

    \pic (leaf4) at (-4.5,0.75) {leafnode={A}};
    \pic (leaf5) at (-1.5,0.75) {leafnode={B}};

    \multiboxnode{i4}{(-2.75,3)}
    \draw[draw=gray!50,thick] (-3.85,2.5) rectangle (-1.75,3.5);

    \draw[bend right] (leaf4-1) to (i4-1-center);
    \draw[bend left] (leaf5-0) to (i4-1-center);

    \draw[bend right] (leaf4-0) to (i4-2-center);
    \draw[bend left] (leaf5-1) to (i4-2-center);

    \multiboxnode{i5}{(0.5,4.5)}[1]
    \draw[draw=gray!50,thick] (-0.2,4) rectangle (1.2,5);

    \draw[bend right] (i4.north) to (i5-1-center);
    \draw[bend left] (i3.north) to (i5-1-center);

    \end{scope}

    \begin{scope}[
      xshift=10cm,
        yshift=4cm
    ]

    \node[vtree]  (v1) at (0,0) {};
    \node[vtree]  (v2) at (-1.5,-1.5) {};
    \node[vtree]  (v3) at (1.5,-1.5) {};
    \node[sink]   (L4) at (-2.25,-3) {$A$};
    \node[sink]   (L5) at (-0.75,-3) {$B$};
    \node[vtree]  (v6) at (0.75,-3) {};
    \node[sink]   (L7) at (2.25,-3) {$E$};
    \node[sink]   (L8) at (0,-4.5) {$C$};
    \node[sink]   (L9) at (1.5,-4.5) {$D$};

    \node[draw=none] (desc) at (-2,-0.5) {$T$:};

    \draw (v1) -- (v2);
    \draw (v1) -- (v3);
    \draw (v2) -- (L4);
    \draw (v2) -- (L5);
    \draw (v3) -- (v6);
    \draw (v3) -- (L7);
    \draw (v6) -- (L8);
    \draw (v6) -- (L9);
      
    \end{scope}

\end{tikzpicture}     }
    \caption{An OBDD and an incomplete TDD computing the same function.}
\end{figure}

Before we formally define TDD, let us discuss the underlying  intuition,
in particular, the  difference of the definition of TDD  
from the circuit based definitions of DNNF variants.
First, the definition of TDD is \emph{explicitly built in} into a vtree in respects.
To put it differently a DNNF may respect some vtree and thus be structured or may not.
For TDD, the respected vtree is not a property it is part of definition. 

The second convention is that the layers of conjunction and disjunction nodes are alternating. 
Further on, the definition explicitly ties each gate to the respected node of the vtree.
For the conjunction nodes this connection is what we can see for structured DNNFs. For a disjunction node $u$,
we introduce an additional requirement that all the conjunction nodes whose outputs are inputs of
$u$ respect the same node of the vtee. So, $u$ respects the very same node.

Possibly, the most remarkable difference  between circuit-based representations and TDD
is that, in TDD, the conjunction nodes are \emph{not explicitly} present!
In particular, we denote by $N(t)$ the set of all the nodes respecting $t$. 
If $t$ is not a leaf then all the nodes of  $N(t)$ are disjunction nodes. 
If $t$ is a leaf then it corresponds to a variable $x$.
In this case, the nodes of $N(t)$ are labelled by one of $x, \neg x, 0,1$. 

The conjunction nodes are \emph{implicitly} present in a TDD in a form of \emph{edges}.
The underlying intuition is as follows.
Let $t$ be a non-leaf node of the respected vtree and let $t_1$ and $t_2$ be the children of $t$. 
Let us fix an arbitrary order on the children assuming $t_1$ to be the left child and $t_2$ being
the right child. Let $u$ be an (implicitly presented) conjunction node of the considered TDD.
Then the inputs of $u$ are $v \in N(t_1)$ and $w \in N(t_2)$. That is $u$ is represented as an
\emph{edge} between $v$ and $w$. 
Now, suppose that $u \in N(t)$. Then the set of conjunction nodes are the inputs of $u$ 
are nothing else than the set of edges that is a subset of $N(t_1) \times N(t_2)$. 
We denote the set $E(u)$. Note that $E(u)$ can also be viewed as a bipartite graph induced by $E(u)$. 

The cornerstone property of TDD that for distinct $u_1,u_2 \in N(t)$ $E(u_1) \cap E(u_2)=\emptyset$,
or to put it differently, the bipartite graphs associated with $u_1$ and $u_2$ are edge disjoint.
This property is critical for the results presented in this paper. 
The other restriction is  that for each leaf $t$, there is the `no simultaneous satisfaction constraint'
for the elements of $N(t)$. 

Our definition of TDD is different from the original one is that we define TDD as \emph{multisink}.
In particular, for the root node $t$ we do not require that $N(t)$ consists of a single node. 
Due to this, we can conveniently assume that for a node $t$ with children $t_1$ and $t_2$ 
the sets $E(u)$ for $u \in N(t)$ form a \emph{partition} of $N(t_1) \times N(t_2)$.  In other words, every pair of nodes $v \in N(t_1)$
and $w \in N(t_2)$, there is a conjunction node with inputs $v$ and $w$ whose output is an input of some node of $N(t)$. 

\begin{definition} \label{def:tdd} 
A Tree Decision Diagram (TDD) $D$ is a 5-tuple $((T,b),N,E, \out, \lambda)$.
In this 5-tuple, $(T,b)$ is a vtree and $N$ is a function mapping each 
$t \in V(T)$ to a set so that for any two distinct $t_1,t_2 \in V(T)$,
$N(t_1) \cap N(t_2)=\emptyset$. 
This means that for each $u \in \bigcup_{t \in V(T)} V(T)$, we can 
unambigously identify $t(u)$ such that $u \in N(t(u))$.

Next, $E$ is a function with the domain $\bigcup_{t \in V(T) \setminus \operatorname{Leaves}(T)} N(t)$.
Let $u$ be an element of the domain, let $t=t(u)$, and let $t_1,t_2$ be the children of $t$
ordered in an arbitrary but fixed way. Then $E(u) \subseteq N(t_1) \times N(t_2)$.
We require that for each $t \in V(T) \setminus \operatorname{Leaves}(T)$, the sets $\{E(u)\mid u \in N(t)\}$
form a partition of $N(t_1) \times N(t_2)$. 

Finally, $\out$ is a set of \emph{outputs} such that $|\out| = N(r)$ where
$r \in V(T)$ is the root of $T$, and $\lambda$ is a bijective function $\out \to N(r)$ assigning each node in the root
bag a unique label.

The semantic of $D$ is described through the functions $D[u]$ computed by
elements $u \in \bigcup_{t \in V(T)} N(t)$. 
Suppose that $t(u) \in \operatorname{Leaves}(T)$ then $D[u] \subseteq \{x,\neg x,0,1\}$
where $x=b(t(u))$, that is, the variable labelling $t(u)$.
We require that for each $t \in \operatorname{Leaves}(T)$ and for each distinct $u_1,u_2 \in N(t)$,
$D[u_1]$ and $D[u_2]$ are not simultaneously satisfiable.
In other words, if $D[u_1]=1$ then $D[u_2]=0$, and $u_1$ and $u_2$ cannot compute the same 
non-constant function. We note that, if $t(u)$ is a leaf then  $D[u]$ is, in fact, part of syntax
as it can be seen as a label assigned to $u$.
If $t(u)$ is not a leaf then $D[u]$ is defined recursively as follows:
$D[u]=\bigvee_{(v,w) \in E(u)} (D[c] \wedge D[w])$.

We refer to $\sum_{t \in V(T)} |N(T)|$ as the \emph{size} of $D$
and denote it by $|D|$.
We also refer to  $\max_{t \in V(T)} |N(T)|$ as the \emph{width} of $D$.

Finally, for every $\ell \in \out$, we also define $D[\ell] \coloneqq D[\lambda(\ell)]$.
\end{definition}

We also consider a special case of TDDs with only two outputs.

\begin{definition}
  A TDD $D$ is called \emph{Boolean}, if $\out(D) = \{0,1\}$. We refer to $D[1]$ as the function computed by $D$.
\end{definition}

Since our notion of Boolean TDDs compute Boolean functions, they are closer to the notion of (regular) TDDs outlined in
the original paper, which have only one output $\out(D) = \{\ell\}$. Furthermore, the original definition and ours differ
in the sense that we demand \emph{completeness}. This means that for each $t \in V(T)$, the sets 
$E(u)$ for $u \in N(t)$ cover the whole $N(t_1) \times N(t_2)$ where $t_1$ and $t_2$ are the children
of $t$. In the original paper it was shown that every TDD can be turned in to a complete TDD
by going bottom-up the tree and adding an additional `zero' element for
each incomplete node $t$ whose $E$-sets includes the elements of $N(t_1) \times N(t_2)$
not covered by the $E$-sets of the other elements of $N(t)$. We can easily verify that by relabelling
the outputs to $0$ and $1$, our notion of Boolean TDDs therefore corresponds to the original notion of complete TDDs.
We chose making TDD multisets and complete so as to streamline our reasoning 
in the next two sections.

An  example of TDD is demonstrated in the central part of Figure 1.
The TDD respects the vtree as depicted on the right-hand side. 
The TDD on Figure 1 is incomplete. In particular, the children of the 
root  are incomplete and the parent of the node labelled by $C$ and $D$
is complete. 

\begin{definition}
    A \emph{multisink-OBDD} (m-OBDD) is an OBDD $B$ in the usual sense, 
    however instead of having two sinks, $B$ has potentially multiple sinks, each labelled with
    a unique output $\ell$. The set of all labels of an m-OBDD $B$ is denoted by $\out(B)$.
    For each $\ell \in \out(B)$, the function $B[l]$ computed by this output is 
    determined by the set of satisfying assignments which is the set of assignments on all paths
    from the source to $\ell$.

\end{definition}

An ordinary OBDD can be viewed as a special case of m-OBDDs, where we have only two labels $0$ and $1$.
The left-hand part of Figure 1 illustrates an ordinary OBDD computing the same function as the TDD illustrated
in the centre of the figure.

Since the multisink models defined above do not compute Boolean functions in the classical sense, we need to define a notion of equivalence as follows.

\begin{definition}\label{lem:corresponding}
    Let $B, D$ be multisink models (either OBDD or TDD). We say that $B$ \emph{corresponds to} $D$, if there is
    a bijection $f: \out(B) \to \out(D)$ such that for all $\ell \in \out(B)$ it holds that $B[\ell] \equiv D[f(\ell)]$.
\end{definition}

In the case of Boolean TDDs, we can use equivalence as usual.
One useful property of TDDs is that, if we restrict an TDD $D$ respecting a vtree $(T,b)$ to a subtree of $T$, we again obtain another well-defined TDD without having to make any other changes. We formalise this notion as follows.

\begin{definition}[Restricted circuits]
    Let $(T,b)$ be a vtree, $D$ an TDD respecting $(T,b)$, and $s \in V(T)$. Then $D_{s}$ is a circuit
    $D'$ respecting $T_s$ such that for all $v \in V(T_s)$ it holds that $N_{D'}(v) = N_D(v)$, and if $v$ is an internal
    vtree node, then $E_{D'}(g) = E_D(g)$ for all $g \in N(s)$. Finally, $\out(D') = N(s)$ such that each node 
    in $N(s)$ is labelled by itself. 
\end{definition} 
\section{Simulating TDDs with OBDDs}

We now want to prove that OBDDs can quasipolynomially simulate TDDs. Our method is inspired by a similar proof that FBDD
simulates decision-DNNF \cite{BeameFBDDSimulation}, however we take a slightly different approach by constructing the resulting OBDD recursively
over the structure of the vtree. In order to accomplish this, we first need to define a single step of our recursion (Lemma \ref{lem:inductionstep}).
For this, we need, in turn a property of TDDs that can be called \emph{unique satisfaction} that immediately follows
from definition.

\begin{lemma}\label{lem:tddproperties}
 Let $(T,b)$ be a vtree with a root $r$ and children $s_1,s_2$.
 Furthermore, let $D$ be a TDD over $(T,b)$ and let $X$ be the set of variables of $D$. 
 Let $\assign$ be an assignment of $X$ and for each $i \in \{1,2\}$, let 
 $X_i$ be the set of  variables of $T_{s_i}$. 

 For each $i \in \{1,2\}$ suppose that there is $\ell_i \in N(s_i)$ 
 such that $\assign_i=\assign \restrict{X_i}$ \emph{activates} $\ell_i$
 (that is, $D[\ell_i](\assign_i)=1$). 
 Then there is exactly one $\ell \in N(r)$ such that 
 $\assign$ \emph{activates} $\ell$. 
\end{lemma}

\begin{lemma}\label{lem:inductionstep}
   Let $(T,b)$ be a vtree with a root $r$ and children $s_1,s_2$. Furthermore, let $D$ be a TDD over $(T,b)$.

   If there are m-OBDDs $B_i$ corresponding to the $D_{s_i}$, then there is an
   m-OBDD $B$ corresponding to $D$ of size $|B_1| + w \cdot |B_2|$, where $w$ is the width of $D$.
\end{lemma}

\begin{proof}
    As each $B_i$ corresponds to $D_{s_i}$, we let $f_i: \out(B_i) \to \out(D_{s_i})$ be the corresponding bijection as in 
    Definition \ref*{lem:corresponding}. Intuitively, we want to simply ``stack'' $B_1$ on top of $B_2$. We can accomplish 
    this by replacing each sink of $B_1$ labelled with $\ell^i$ with a copy of $B_2$ with sinks labelled with
    $\ell_1^i, \dots \ell_{L}^i$, where $L \coloneqq |\out(B_2)|$. Let $B'$ be the resulting m-OBDD.

    Let $X_1$ be the variables of $B_1$, and $X_2$ be the variables of $B_2$. We observe that an assignment $\assign$
    of $X_1 \cup X_2$ is a model of  $B'[\ell_j^i]$ exactly if $\assign \restrict{X_1}$ is a model of $B_1[\ell^i]$,
    and $\assign \restrict{X_2}$ is a model of $B_2[\ell_j]$. Since each $B_i$ corresponds to $D_{s_i}$, it follows
    that $\assign \models D_{s_1}[f_1(\ell^i)]$, and $  \assign \models D_{s_2}[f_2(\ell_j)]$. 
    By the unique satisfaction property, there must
    also be exactly one $g \in N(r)$ such that $\assign$ activates $g$. This means that we can take all $\ell_j^i$ such that the corresponding
    pair $(f_1(\ell^i), f_2(\ell_j))$ belong to the same $E(g)$, and merge them into a single sink labelled with $g$. We let
    $B$ be the resulting m-OBDD.

    It clearly follows that every assignment $\assign$ of $B$ reaches a sink labelled with $g$ exactly if $\assign$
    activates $g$ in $D$. Therefore, $B$ corresponds to $D$. For the size bound, we simply observe that $|\out(B_2)| \leq w$,
    where $w$ is the width of $D$, and that this construction only requires one copy of $B_1$, and $|\out(B_2)|$ copies of $B_2$.
\end{proof}

Now that we have completed a single step of our transformation, we only need to show that we can apply Lemma \ref*{lem:inductionstep} in a way that maintains our desired upper bound.

\begin{theorem}\label{th:simulation}
    Let $(T,b)$ be a vtree, and $D$ be a TDD over $(T,b)$. Then there is a linear order $\leq$, as well as
    an m-OBDD $B$ respecting $\leq$ and corresponding to $D$ of size $\bigO\left( (\sqrt{n} \cdot w(D))^{\log(n)+2} \right)$,
    where $w(D)$ is the width of $D$.
\end{theorem}

\begin{proof}
    The idea is to iteratively apply Lemma \ref*{lem:inductionstep} in order to obtain the desired m-OBDD. The correctness of the
    construction immediately follows, however we still have to make sure that the size bound is correct. Towards this goal,
    we want to add some more structure to our vtree.

    For every $s \in V(T)$, we let $|s|$ be the number of leaves in $T_s$. Furthermore, let $s \in V(T)$ be some internal node
    of $T$ with children $t_1, t_2$. We now label $t_i$ as the \emph{left child} of $s$, if $|t_i| > |t_{3-i}|$,
    and label $t_{3-i}$ as the \emph{right child} of $s$. If $|t_1| = |t_2|$, we assign the left and right child arbitrarily.
    The linear order $\leq$ is then defined recursively with $x_1 < x_2$ for all variables $x_1$ mentioned in the left child, and $x_2$ in the right child of $s$.

    We now apply Lemma \ref*{lem:inductionstep} bottom-up over the vtree as follows: For every leaf $p$ of $T$, we can trivially
    define an m-OBDD $B_p$ corresponding to $D_{p}$. For every internal node $s$ of $T$ with children $t_1$,
    and $t_2$, if we have already constructed the m-OBDDs $B_{t_i}$ corresponding to both $D_{t_i}$, then we apply Lemma \ref*{lem:inductionstep}
    with $t_1$ as the left child of $s$. This means that we only ever create copies of the m-OBDD $B_{t_2}$.

    Let's take a closer look at this construction, and consider one inductive step of our construction for 
    some internal node $s$ with children $t_1, t_2$. Let $k = |s|$, and for $i \in \{1,2\}$,
    let $k_i = |t_i|$. Since the variables mentioned by either $t_i$ are disjoint, and $k_1 \geq k_2$, we observe that 
    $1 \leq k_2 \leq \lfloor \frac{1}{2} k \rfloor$ as well as $k_1 \leq k-1$.

    We now want to study the size of the smallest m-OBDD $B_s$ coresponding to $D_{s}$ relative to $k$, as well as $k_1$ and $k_2$.
    In particular, we let $M[k]$ be the maximum size among the constructed m-OBDDs $B_p$ corresponding to $D_{p}$ for any 
    $|p| \leq k$. This also means that $M[n]$ is an upper bound of the size of an m-OBDD corresponding to $D$. By Lemma \ref*{lem:inductionstep}, 
    the size of each of our $B_s$ is $|B_s| \leq |B_{k_1}| + w \cdot |B_{k_2}|$, where $w$ is the width of $D$. By the definition of 
    $M[k]$ and our previous observation about $k_1$ and $k_2 $ it follows that $B_{k_1} \leq M[k-1]$, and 
    $B_{k_2} \leq M[ \left\lfloor \frac{k}{2} \right\rfloor]$. We conclude that
    
    \[ M[k] \leq M[k-1] + w \cdot M[ \left\lfloor \frac{k}{2} \right\rfloor  ]. \]

    Let $r$ be the root of the vtree $(T,b)$. It clearly holds that $D_{r} = D$. Furthermore it holds by definition of $M[n]$ 
    that $|B_r| \leq M[n]$, where $B_r$ is the constructed m-OBDD corresponding to $D_{r}$. Therefore, it is 
    sufficient to show that $M[n]$ satisfies the desired upper bound in order to prove the theorem.

    We simplify this recurrence by iteratively expanding on the linear part. This gives us $M[k] \leq 2w \cdot \sum_{i=1}^{\frac{k}{2}}M[i] + M[1]$.
    Because $M[i]$ is a monotone function, we have that $M[i] \leq M[\frac{k}{2}]$ for each $i \leq \frac{k}{2}$. If we additionally assume for simplicity that $k = 2^\ell$ for some $\ell \in \N$, we get $M[k] \leq C[k]$, where $C[k]$ is a reccurence defined as
    \[ C[k] \coloneqq wk \cdot C[\frac{k}{2}] + C[1].   \]
    
    We now let $N[k] \coloneqq wk \cdot N[\frac{k}{2}]$, and $F(k) \coloneqq C[k] - N[k]$ be the sum of all remaining
    $C[1]$. It clearly holds that $C[k] = N[k] + F(k)$, therefore, we only need to show that both of these functions satisfy our upper bound. 
    We start by analysing $N[k]$:
    \begin{equation*}
        \begin{aligned}
            N[k] &\leq wk \cdot N[\frac{k}{2}] = \prod_{i =0}^{\ell} wk \cdot 2^{-i} = (wk)^{\ell+1} \cdot \prod_{i =0}^{\ell} 2^{-i}  \\
            &\leq (wk)^{\ell+1} \cdot 2^{- \frac{\ell(\ell+1)}{2}} = (wk)^{\ell+1} \cdot k^{- \frac{\ell+1}{2}} = w^{\ell+1} \cdot k^{\frac{\ell+1}{2}}  \\
            &= (w \sqrt{k})^{\ell+1} \in \bigO\left( \left( w \sqrt{k} \right)^{\ell +2}\right).
        \end{aligned}
    \end{equation*} 
    We recall that $\ell = \log k$, therefore $N[k]$ satisfies the desired upper bound.
    It remains to show that $F(k)$ is in $\bigO\left( (w \sqrt{k})^{\ell +2} \right)$ as well.
    Towards this goal, we want to expand $C[k]$ on each instance of $C[\frac{k}{2}]$, and then split $F(k)$ into several $F_i(k)$ with each representing one of our expansion steps. We observe that during the $i$th expansion the recurrence $C[k]$ takes the form of

    \[ C[k] = \prod_{j = 0}^{i} \left( wk \cdot 2^{-j} \right) C\left[  \frac{k}{2^{i+1}} \right] + \sum_{j = 0}^i F_i(k), \]
    
    For example, for $i = 0$, this expansion is just the regular definition of $C[k]$ with $F_0(k) = C[1]$, and for $i=1$, we have $C[k] = wk \cdot w\frac{k}{2} C[\frac{k}{4}] + C[1] + wk \cdot C[1] $ with $F_1(k) = wk \cdot C[1] = wk \cdot 2^{-0}C[1]$, which is just one expansion on $C[\frac{k}{2}]$. In general we observe that $F(k) = C[1] + \sum_{i = 1}^{\ell} F_i(k)$ with
    
    \[ F_i(k) = C[1] \cdot \prod_{j =0}^{i-1} wk  \cdot 2^{-j}. \]
    We can now simplify this expression to
    \begin{equation*}
        \begin{aligned}
            F_i(k) &= C[1] \cdot (wk)^{i} \prod_{j = 0}^{i-1} 2^{-j} = C[1] \cdot (wk)^i \cdot 2^{-\frac{(i-1)(i)}{2}} = C[1] \cdot \left( wk2^{-\frac{i-1}{2}} \right)^i.
        \end{aligned}
    \end{equation*}
    We observe that $F_i(k) < F_{i+1}(k)$ for all $i < 2\ell$, as $2^{\frac{i-1}{2}} < 2^\ell = k$, and for $i = \ell +1$ it holds that $2^{\frac{i-1}{2}} = 2^{\frac{\ell}{2}} = \sqrt{k}$.
    Therefore, for all $i \leq \ell$ it holds that $F_i(k) \leq F_{\ell +1}(k) =  C[1] \left( w\sqrt{k} \right)^{\ell +1}$. We finally conclude

    \begin{equation*}
        \begin{aligned}
        F(k) &= C[1] + \sum_{i =1}^{\ell} F_i(k) \leq C[1] + \sum_{i =1}^{\ell} C[1] \cdot \left( w\sqrt{k}  \right)^{\ell +1} \\
        &= C[1] + \ell \cdot C[1] \left( w\sqrt{k}  \right)^{\ell +1} \\
        &\leq C[1] + C[1] \left( w\sqrt{k}  \right)^{\ell+2} \in \bigO\left( \left( w \sqrt{k}  \right)^{\ell +2} \right)
        \end{aligned}
    \end{equation*}
    Since both $N[k]$ and $F(k)$ satisfy the desired upper bound, our claim follows.
\end{proof}

\section{Restructuring TDDs}

Another interesting problem for OBDDs is the so-called \emph{global rebuilding problem}. The goal of this problem is to take any OBDD $B$ respecting a linear order $\leq$, and build an equivalent reduced OBDD $B'$ respecting a new linear order $\leq'$. It was shown in \cite{WegenerBranchingPaths} that this problem can be solved in time $\bigO(|B||B'| \cdot \log(|B'|)$. Since OBDDs support polynomial time equivalence checking over the same linear order, it follows that they also support polynomial time equivalence checking over \emph{different} linear orders.

In the context of TDDs we are instead interested in arbitrarily changing the underlying vtree, therefore we are going to refer to this problem as the \emph{restructuring problem}.
The goal of this section is to show that TDDs admit polynomial-time restructuring. And as a consequence we can also show that TDDs not only support polynomial-time equivalence checking over different vtrees, but we can also efficiently check equivalence between TDDs and much more powerful data structures such as  deterministic structured decomposable negation normal form (d-SDNNF) \cite{darwichedsdnnf}.

We start by introducing the notion of \emph{factors}, which play a key role in minimising TDDs. Factors were originally introduced in \cite{DBLP:conf/pods/BovaS17}, however we use a slightly different equivalent definition.

\begin{definition}
    Let $f: \{0,1\}^X \to \{0,1\}$ be a Boolean function, and $Y \subseteq X$. Two partial assignments 
    $\assign_1, \assign_2: Y \to \{0,1\}$ are called \emph{equivalent}, denoted by $\assign_1 \equiv_Y^f \assign_2$, 
    if $f \restrict{\assign_1} \equiv f \restrict{\assign_2}$.
    A \emph{factor} of $f$ (with respect to $Y$) is an equivalence class of $\equiv_Y^f$.
\end{definition}

We note that by definition factors always inclusion-wise maximal. It was shown in \cite{tdd} that TDDs can be characterised 
using factors in the following sense:

\begin{definition}
  Let $D$ be a Boolean TDD over a vtree $(T,b)$. $D$ is called \emph{reduced}, if for every $t \in V(T)$
  and $u \in N(t)$ it holds that $D[u]$ is a unique factor of $D$ with respect to $X_{t}$, where
  $X_t$ is the set of variables mentioned in $T_t$.
\end{definition}

It was further shown that for every node $u$ of $D$ all assignments in $D[u]$ are equivalent, and therefore every TDD
(in the sense of \cite{tdd}) can be reduced in polynomial-time by simply merging equivalent nodes, and that this reduced
form is always canonical. We further observe that, if we turn a reduced TDD into a complete TDD, as outlined in Section \ref*{sec:prelim},
then the resulting complete TDD is also minimal and canonical among all complete TDDs respecting the vtree $(T,b)$.
We recall that our notion of Boolean TDDs corresponds to the original notion of complete TDDs. We therefore conclude:

\begin{lemma} \label{lem:reducedtdd}
    Let $D$ be a Boolean TDD. There is a polynomial-time algorithm that computes a reduced Boolean TDD $D'$
    such that $D'$ is equivalent to $D$. Furthermore, this reduced TDD is minimal and canonical. 
\end{lemma}

For convenience, we are only going to consider Boolean TDDs throughout this section. However, it can be shown that the algorithm
in Lemma \ref*{lem:reducedtdd} also works on regular TDDs, and always outputs a minimal and canonical form.

We are now ready to prove our main theorem. We start by proving that TDDs support polynomial-time restructuring.

\begin{theorem} \label{th:restructuring}
    Given a Boolean TDD $D$ and a target vtree $(T,b)$, there is an algorithm that can construct a reduced TDD $D_{(T,b)}$ over $(T,b)$ such that
    $D_{(T,b)} \equiv D$ in time $|D_{(T,b)}|^4 \cdot |D|^{\bigO(1)}$.
\end{theorem}

\begin{proof}
    The idea is that we construct $D_{(T,b)}$ bottom-up over the structure of the vtree, and at each step use the original $D$ in order
    to check whether any two nodes of $D_{(T,b)}$ belong to the same factor. For this, we store for each vertex $v \in V(T)$
    the set $N(v)$ as well as a representative model $\assign_g$ for every gate $g$ in $N(v)$. The process works as follows:

    For every leaf $\ell$ of $(T,b)$, check whether $D \restrict{b(\ell)=0} \equiv D \restrict{b(\ell) = 1}$. 
    If yes, then
    $N(\ell)$ is just a single node labelled with $\top$ and we pick an arbitrary assignment as the representative model, 
    otherwise $N(\ell) \coloneqq \{ g_1 \coloneqq \{b(\ell) \}, g_0 \coloneqq \{\neg b(\ell)\} \}$, and we let $\assign_{g_i}$ 
    be the unique model of the gate $g_i$.

    Let $s$ be an internal node of $(T,b)$ with children $t_1,t_2$ such that we have already derived $N(t_1)$ and $N(t_2)$.
    We construct $N(s)$ as follows:
    
    Intuitively, we first add one node $g_{ij}$ for each $(g_i,g_j) \in N(t_1) \times N(t_2)$, and then merge all
    pairs of nodes, which belong to the same factor. We start by defining $N'(s) \coloneqq \{g_{ij} \mid (i,j) \in N(t_1) \times N(t_2)\}$
    with $E(g_{ij}) = \{ (i,j) \}$ and representative model $\assign_{g_{ij}} \coloneqq \assign_i \cup \assign_j$ for all pairs $(i,j)$.
    Next we want to merge all nodes belonging to one factor. Since all models of a node $g$ are equivalent,
    we can check whether two nodes belong to the same factor by only comparing their representative models.
    For all $g,g' \in N'(s)$ we define the equivalence relation $g \equiv^D g'$, exactly if the corresponding $D \restrict{\assign_g} 
    \equiv D\restrict{\assign_{g'}}$. Let $\mathcal{C}_D$ be the equivalence classes of $\equiv^D$.

    We observe that two nodes $g$ and $g'$ belong to the same equivalence class of $\equiv^D$ if, and only if, they belong to the same factor,
    and thus can be merged into a single TDD-node. Therefore we define $N(s)$ as the set of all $g_C$ for every 
    equivalence class $C \in \mathcal{C}_D$ with $E(g_C) \coloneqq \bigcup_{g \in C} E(g)$, and as a representative model 
    $\assign_C$ we pick the representative model $\assign_g$ of an arbitrary $g \in C$. Finally, at the root $r \in V(T)$,
    We label every $g \in N(r)$ with $i$, if $\assign_{g}$ is a model of $D[i]$, where $i \in \{0,1\}$ is a possible 
    output of $D$.

    For the correctness of the algorithm, it is sufficient to recognise that by induction every $N(s)$ describes the factors 
    of $D$ with respect to $X_s$. Therefore, if $s$ has children $t_1,t_2$, then all models $t_1$ and $t_2$ belong to the same factor.
    It follows that the models of every $(g_1,g_2)$ with $g_1 \in N(t_1)$ and $g_2 \in N(t_2)$ belong to the same factor. In the last
    step we only merge pairs in $N'(s)$, if the representative models belong to the same factors, it follows immediately that 
    $N(s)$ describes a set of factors as well.

    For the size bound, we essentially minimise $D_{(T,b)}$ ad hoc, therefore the resulting TDD will always be reduced. The size of any 
    intermediate circuit is therefore $|D_{(T,b)}| + w^2$, where $w$ is the width of $D_{(T,b)}$, as every $N(t_1)$ and $N(t_2)$
    has size at most $w$, and the largest intermediate addition is $|N'(s)| = |N(t_1)| \cdot |N(t_2)| \leq w^2$.

    Finally, the equivalence classes of $\equiv^D$ can be determined in time $w^4 \cdot |D|^{\bigO(1)}$, as $N'(s)$
    has size $w^2$, and there are $|N'(s)|^2$ pairs in $N'(s)$. The time to check each pair is just the time to check
    equivalence between the $D \restrict{\assign_g}$ and $D\restrict{\assign_{g'}}$. This can be done in polynomial time, because in both cases we can  compute the canonical minimal TDD of the restriction in polynomial time \cite{tdd}. Since they both respect the same vtree, they are equivalent if and only if they are equal.
\end{proof}

In order for this proof to work out, we do not need $D$
to necessarily be a TDD. Indeed, $D$ could be any representation format, as long as it supports polynomial time restriction to a partial assignment and polynomial time equivalence checking
of two forms over the same vtree. Therefore, we can easily extend this theorem to much stronger data structures, such as sentential decision diagrams (SDDs) \cite{darwiche2011sdd} and deterministic structured decomposable negation normal form (d-SDNNF) \cite{darwichedsdnnf}, both of which are more succinct than TDD, but support polynomial-time equivalence checking over the same vtree.

 \begin{corollary}\label{cor:tdd-dsdnnf}
    Given an OBDD/SDD/d-SDNNF $D$ and a target vtree $(T,b)$, there is an algorithm that can construct a Boolean TDD $D_{(T,b)}$ over $(T,b)$ such that
    $D_{(T,b)} \equiv D$ in time $|D_{(T,b)}|^4 \cdot |D|^{\bigO(1)}$.
 \end{corollary}

 We can easily use this restructuring method to check equivalence between TDDs over different vtrees.

 \begin{lemma}\label{lem:tdd-different-vtrees}
    Let $D,D'$ be Boolean TDDs over different vtrees $(T,b), (T',b')$. It is possible to check whether $D \equiv D'$ in polynomial time.
 \end{lemma}

 \begin{proof}
    We start by reducing both $D$ and $D'$. Next, we simply run the algorithm from Theorem \ref*{th:restructuring} on $D$ and $(T',b')$, 
    and let the resulting TDD be called $\hat{D}$. It is now sufficient to check whether $\hat{D} \equiv D'$. 
    Since $\hat{D}$ and $D'$ are defined over the same vtree, this final check can be done in polynomial time. It might be possible that,
    if $\hat{D}$ is not equivalent to $D'$, then $\hat{D}$ might be superpolynomially larger than $D'$. However, since we always
    minimise $\hat{D}$ ad hoc, we can immediately reject whwenever any intermediate result is larger than $|D'|^2$, as this is the largest
    possible blow-up that can occur during the restructuring process.

    The correctness of this construction is clear, as $\hat{D}$ is equivalent to $D$. Therefore, $D \equiv D'$ exactly if $\hat{D} \equiv D'$,
    and the runtime is polynomial since we either construct $\hat{D}$ in polynomial time, or reject immediately, whenever the circuit blows up too much.
 \end{proof}

 We can also slightly alter the method described in Theorem \ref*{th:restructuring} in order to construct an equivalence checking method
 between TDDs and much more powerful representation formats.

 \begin{lemma}\label{lem:eqtdddsdnnf}
    Let $D$ be a Boolean TDD over a vtree $(T,b)$, and $D'$ a d-SDNNF respecting an arbitrary vtree. then it is possible to check
    whether $D \equiv D'$ in polynomial time. 
 \end{lemma}

\begin{proof}
    We can simply use Corollary \ref*{cor:tdd-dsdnnf} in order to construct a TDD $\hat{D}$ equivalent to $D'$ over the vtree $(T,b)$,
    and then we can simply check whether $D$ and $\hat{D}$ are equivalent in polynomial time. Using the same argument as in Lemma
    \ref*{lem:tdd-different-vtrees}, we can immediately reject, if any intermediate step during the construction of $\hat{D}$ grows too large.
\end{proof} 
\section{Conclusion}

In this paper, we have shown two interesting properties regarding TDDs: first we have shown that every TDD can be simulated by an OBDD with only quasipolynomial blowup. And second, we have shown that TDDs admit efficient restructuring, which opens up several more nice properties such as polynomial-time equivalence checking between TDDs and data structures such as d-SDNNFs.
With that said, there are still several open research questions: first of all, can either restructuring in time polynomial in the input and output or the polynomial time equivalence test be lifted to SDDs or even all d-SDNNFs? 

Note that for SDDs the canonical form is not of minimal size, whereas canonical forms of TDDs (and OBDDs) are minimal. On the other hand, SDDs are exponentially more succinct than OBDDs, whereas TDDs are not (as we now know). Can we have both? That is, does there exists a structured representation format with canonical minimal forms (for a given vtree and a Boolean function) that are exponentially more succinct than OBDDs?

Finally, as it has been demonstrated that TDDs are quite efficient  for model counting, it would be good to further explore how well TDDs perform in praxis, in particular in the context of query evaluation.

\bibliography{refs}
\nocite{*}
\end{document}